\documentclass[twoside,11pt]{article}

\usepackage{jmlr2e}
\usepackage[dvipsnames]{xcolor}
\usepackage{amsmath}
\usepackage{mdframed}
\usepackage[margin=1.0in]{geometry}
\usepackage{longtable}

\ShortHeadings{Optimal Grain Mixing is NP-Complete}{Noor, Yaw, Zhu, and Sheppard}
\firstpageno{1}

\begin{document}

\title{Optimal Grain Mixing is NP-Complete\\Tech Report \#MSU-NISL-21-001}

\author{\name Md Asaduzzaman Noor
    \email mdasaduzzaman.noor@student.montana.edu\\ 
    \addr Gianforte School of Computing\\
    Montana State University\\
    Bozeman, MT 59717
\AND 
        \name Sean Yaw 
    \email sean.yaw@montana.edu\\
    \addr Gianforte School of Computing\\
    Montana State University\\
    Bozeman, MT 59717
\AND 
        \name Binhai Zhu 
    \email bhz@montana.edu \\
    \addr Gianforte School of Computing\\
    Montana State University\\
    Bozeman, MT 59717
\AND
        \name John W. Sheppard 
    \email john.sheppard@montana.edu \\
    \addr Gianforte School of Computing\\
    Montana State University\\
    Bozeman, MT 59717
}

\maketitle

\begin{abstract}
Protein content in wheat plays a significant role when determining the price of wheat production. The Grain mixing problem aims to find the optimal bin pair combination with an appropriate mixing ratio to load each truck that will yield a maximum profit when sold to a set of local grain elevators. In this paper, we presented two complexity proofs for the grain mixing problem and showed that finding the optimal solutions for the grain mixing problem remains hard. These proofs follow a reduction from the $3$-dimensional matching ($3$-DM) problem and a more restricted version of the $3$-DM known as planar $3$-DM problem respectively. The complexity proofs do suggest that the exact algorithm to find the optimal solution for the grain mixing problem may be infeasible. 
\end{abstract}


\section{Introduction}
Generally, agriculture and agricultural products are essential in sustaining lives on the planet. Considerable planning is required in the agricultural sector To feed the large population on the Earth efficiently. In this paper, we considered an important component in the wheat distribution profit referred to as grain mixing (wheat blending). When selling wheat in a local grain elevator, many variables come into play for determining the price of the wheat. Among them, protein content plays the most important role which is affected by several environmental factors such as temperature, soil nitrogen level, precipitation, etc. Due to these variances, protein content in wheat changes not only from year to year but also from crop to crop. 

After the harvesting season, usually, the farmers store their grain into multiple grain bins and transport the grain via trucks in batches to sell wheat in the local grain elevators. Today, the device for tracking protein content in wheat is available, however, it is expensive. Therefore, small farmers take their harvest to the nearest elevator and collect the price paid by them. The grain mixing problem aims to determine the optimal mixing of different quality wheat (in terms of protein content) to load trucks that maximize the overall profit when selling wheat to multiple local grain elevators. 

This paper extends our previous work (\cite{noor_21}) where we applied and adapted two different evolutionary approaches: genetic algorithms and differential evolution in order to solve the grain mixing problem. The experimental results suggested that mixing grain increases the profitability when selling wheat, and the evolutionary approaches consistently led to higher overall profit. In this paper, we investigated the complexity proof for the grain mixing problem. We showed that the grain mixing problem is NP-Complete following a reduction from the ($3$-DM) and a planar ($3$-DM) problem respectively. Both of these problems are known to be NP-Complete (\cite{complexity_90, 3dm_npc}). The complexity proofs justify the use of approximation algorithms (such as evolutionary approaches) for getting a feasible solution. 

\section{Background}
There have been several approaches in the literature that attempt to solve the wheat blending problem (or blending problems in general). A few works in the literature attempted to solve the decision version of the wheat blending problem using linear programming (LP). \cite{blending_2001} utilized LP methods capable of predicting the optimal wheat blend ratio for a targeted final quality to produce a bread-making flour. \cite{blending_thesis} used the simplex algorithm to find the optimum blend that satisfies the customer's specific solvent retention capacities (SRC). However, for the grain mixing problem studied here, the protein cost function is non-linear, and there is no targeted wheat quality (the protein content of a truck is determined in runtime). Therefore, it can not be applied directly to the LP models. 

Mixed Integer Linear programming (MILP) is often used to solve real-world blending problems with problem-specific constraints. Although the MILP model can be used to get exact solutions, it is known to be NP-Hard (\cite{milp_hardness_88}). \cite{MILP_blending_shipping} proposed an MILP model to optimize the cost for the wheat supply chain (blending, loading, transportation, and storage), where the model used a specific blending formula for mixing. Meta-heuristic approaches are also a popular choice for solving blending optimization problems. \cite{blending_hybrid_evo} proposed a hybrid evolutionary method to solve the wheat blending problem in Australia. Their problem closely relates to ours, however, there are some additional constraints in our problems based on the US wheat market. The real-world blending problems in the literature do suggest that an exact solution may not be feasible in many cases. To the best of our knowledge, there is no complexity proof for the grain mixing (wheat blending) problem.

\section{Grain Mixing Problem}
\subsection{A Simple Example}
First, we will demonstrate how mixing grain can improve profitability when selling wheat through a simple example. We collected wheat harvesting data from a local Montana farmer who tracks the protein content of his wheat. In Montana, when selling wheat, the price per bushel of grain depends on a range of protein content. Each elevator has a base protein content range for which a base price is paid. Then the cost model follows a premium-dockage curve where if a tuck's protein content is higher than the base protein range a premium price is paid based on how far it is from the base level, and the price is docked if the truck's protein level is lower than the base protein range. In many cases, the protein content in a bin is short for reaching a higher price range. Therefore, mixing it with a high protein content bin when loading a truck might change the average protein content of the truck to reach a higher price grade which provides the scope for optimization. 

Figure \ref{fig:backgroud_grain_mix_example} shows an example of the grain mixing problem. There are three bins with different amounts of bushels and protein content in the example. Table \ref{tab:backgroud_protein_table} shows the example elevator prices where $[11, 12)\%$ shows the base protein range with $\$4$ price. The price is docked for a lower protein range and increased for a higher protein range. Farmers who do not track the protein content of wheat would load trucks separately with grain from each bin and the price they would get would be $(50 * \$3 + 100 * \$4 + 50 * \$6) = \$850$. There is a fixed capacity for the number of bushels a truck can carry which is $100$ in the example. However, if they were to track the protein content and mix grains as shown in the figure; load truck one with $50$ bushels from bin one and $50$ bushels from bin two, and load truck two with $50$ bushels from bin two and $50$ bushels from bin three, the price they will get would be $(100 * \$4 + 100 * \$6) = \$1000$. Therefore, mixing grain in this scenario increases the profit by $\$150$. 

\begin{figure}[t!]
    \centering
    \includegraphics[width=0.6\linewidth]{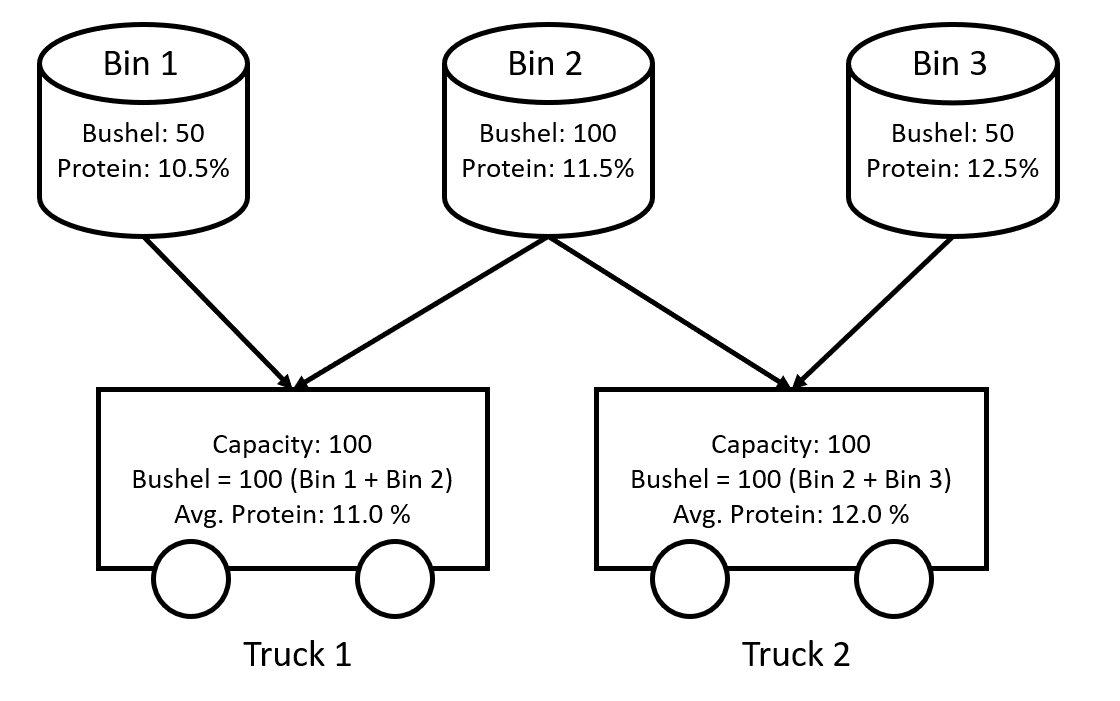}
    \caption{A simple grain mixing example}
    \label{fig:backgroud_grain_mix_example}
\end{figure}
\begin{table}[t!]
\centering
    \caption{Elevator price for grain mixing example}
    \begin{tabular}{|c|c|}
    \hline
    \textbf{Protein range (\%)} & \textbf{Price/Bushel (\$)} \\ \hline
    $[10, 11)$ & $3$ \\ \hline
    $[11, 12)$ & $4$ \\ \hline
    $[12, 13)$ & $6$ \\ \hline
    \end{tabular}
    \label{tab:backgroud_protein_table}
\end{table}

Besides protein content, two more variables need to be considered in the cost model. There is a mixing cost associated when mixing grain to load a truck. The mixing cost depends on the distance between the two bins that were used to load a truck. The mixing cost is higher when mixing grain from two bins that are farther apart. Finally, there is a delivery cost associated with the distance between the truck site and the elevators. 

A key challenge in the grain mixing problem is to find the optimal bin pair combination and the number of bushels drawn from each bin to load tucks that will yield maximum profit. In this paper, we showed that finding the optimal bin pair combination with the mixing ratio is an NP-Complete problem. 

\subsection{Problem Definition}
In order to derive the complexity proofs from the $3$-DM problems, we had to assume that each elevator has a fixed amount of bushels that they will accept. Although this assumption is not present in the general grain mixing problem, the proof provides us an idea of the hardness of the grain mixing problem studied here. Therefore, the grain mixing problem can be defined as follows:

\noindent Input:
\begin{itemize}
    \item $B$: Set of bins with each having a capacity, protein content (\%), and elevator-specific delivery cost. Each pair of bins also has a mixing cost.
    \item $R$: Set of trucks with each having a capacity.
    \item $M$: Set of elevators with each having a capacity and protein content cost function (i.e., \$ paid per unit with protein content $p$).
\end{itemize}

The goal of the Grain Mixing ($GM$) problem is to find an assignment of bins to trucks to elevators and quantity of grain to transfer such that:
\begin{enumerate}
    \item No bin, truck, or elevator capacities are violated.
    \item The total profit is maximized.
\end{enumerate}

\section{Complexity Proofs}
The complexity proofs for the grain mixing (GM) problem follow from a reduction from the standard $3$-DM problem and a more restricted version of the $3$-DM problem known as planar $3$-DM problem respectively. The standard $3$-DM problem can be defined as follows:

\begin{definition}[$3$-DM]
Let $X$, $Y$, $Z$ be finite sets with $|X| = |Y| = |Z| = \alpha$, and let $T \subseteq X \times Y \times Z$ consist of triples such that $x \in X$, $y \in Y$, and $z \in Z$. $M \subseteq T$ with $|M| = \alpha$ is a valid 3-Dimensional Matching (3-DM) if for any two distinct triples $(x_1, y_1, z_1) \in M$, and $(x_2, y_2, z_2) \in M$, the following holds: $x_1 \neq x_2$, $y_1 \neq y_2$, and $z_1 \neq z_2$.
\end{definition}

Similarly, from the standard $3$-DM problem, the restricted planar $3$-DM can be defined as follows:
\begin{definition}[Planar $3$-DM]
Let us consider a standard $3$-DM problem. A bipartite graph can be associated with the 3-DM instance where one side consists of vertices from $X$, $Y$, and $Z$ and the other side consists of vertices from triples $T$. There is an edge between the vertices from the sets to the vertices of the triples if and only if the set element is in the triple. If the associated graph is planar then it is referred to as planar 3-DM.
\end{definition}

The decision version of the (planar) $3$-DM problem is known to be NP-complete as shown in one of Karp's NP-Complete problem lists (\cite{Karp1972}). The decision problem can be defined as \textit{``Given a subset of triples $T$ and an integer $\alpha$, does there exist a $3$-dimensional matching $M \subseteq T$ where $|M| \geq \alpha$''}. Consequently, the optimization problem for the $3$DM can be defined as \textit{``Given a subset of triples $T$, find the $3$-dimensional matching $M \subseteq T$  that maximizes $|M|$''}. As the decision problem is NP-Complete, that follows that the optimization problem will be NP-Hard.

\begin{theorem}
\label{thm:nphard}
The grain mixing problem is NP-Complete.
\end{theorem}
\begin{proof}
We provide two proofs for the NP-Completeness of the GM problem. The first one follows a reduction from the planar $3$-DM problem. First, given an instance of planar 3-DM problem, we convert it to an instance of GM problem as follows:
\begin{itemize}
    \item For each element in set $X$, and set $Y$:
    \begin{itemize}
        \item Create Bin $b_{x_i}$ and Bin $b_{y_z}$ and assign half unit of bushels,
        \item Set protein content for $b_{x_i}$ to $(p-\epsilon) \%$ and for $b_{y_z}$ to $(p+\epsilon) \%$,
        \item Set initial mixing cost to $+\infty$ for all of the created bins.
    \end{itemize}
    \item For each element in set $Z$, create an elevator $m_{z_k}$ and define its cost function to pay $\$0$ initially. Set the capacity of the elevators to accept a maximum of one unit bushel.
    
    \item For each triples $(x_i, y_j, z_k) \in T$, create a truck $t_i$ that loads grain from $b_{x_i}$ and $b_{y_z}$ and deliver to elevator $m_{z_k}$. The maximum truck capacity is also one unit bushel.
    \begin{itemize}
        \item change the mixing cost between $b_{x_i}$ and $b_{y_z}$ to $C$,
        \item Set the delivery cost of truck $t_i$ to $C$,
        \item Update elevator $m_{z_k}$ cost function to pay $\$R$ if the protein range of the truck is in $\left[p, p+2\epsilon\right)$.
    \end{itemize}
    
\end{itemize}

Figure \ref{fig:gm_proof} shows the reduction from the planar 3-DM to the GM instance. In the figure, the triples represented in the red box could be considered as an intermediate warehouse that is connected to the elevator by an edge with cost $+\infty$ that is used to prevent flowing trucks/bushels from one elevator to another (as it is not possible for the GM scenario). The triples vertex with the associated set vertex can be thought of as a truck starting at $b_{x_i}$ location and going to $b_{y_j}$ location for mixing grain incurring mixing cost $C$ and from there going to the warehouse with a delivery cost of $C$. These costs are represented with the edge distance. We then prove the following claim:  {\it The planar 3-DM problem has a matching of size $\alpha$ if and only if the GM instance has a maximum revenue of $\$R\times \alpha$ and a minimum cost of $2C \times \alpha$.}

\begin{figure}[t!]
    \centering
    \includegraphics[width=0.8\linewidth]{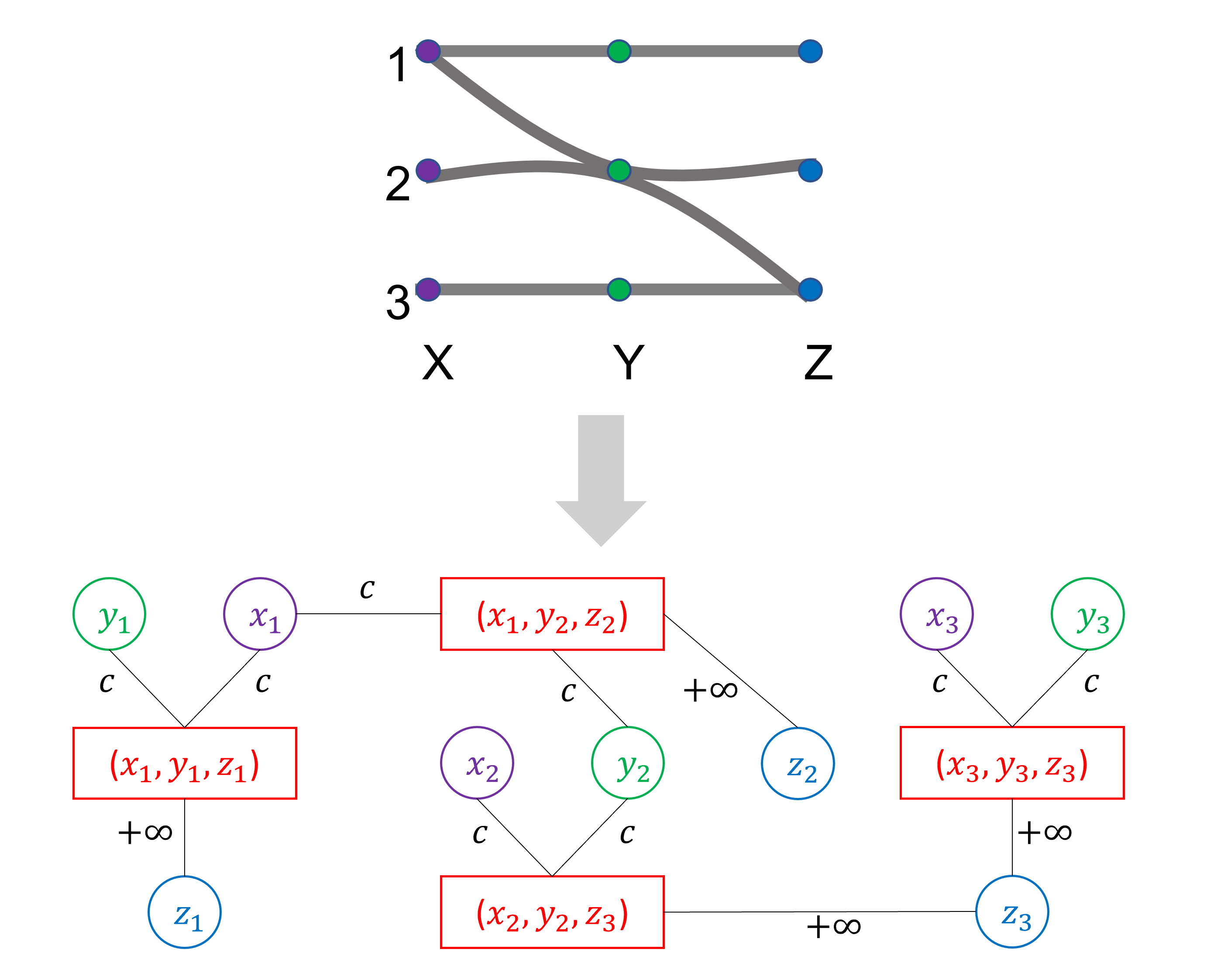}
    \caption{Reduction from planar 3-DM to GM instance}
    \label{fig:gm_proof}
\end{figure}

\vspace{0.2in}
\noindent
{\bf If part: }{\it If the planar 3-DM has a matching of size $\alpha$ then the GM instance a maximum revenue of $\$R\alpha$ and a minimum cost of $2C\alpha$.}

For each matching in triples $(x_i, y_j, z_k) \in \alpha$, we load trucks with half unit bushels from bin $b_{x_i}$, and half unit from bin $b_{y_j}$ and send to elevator $m_{z_k}$. The elevator will pay $\$R$ as the protein range of the truck would be $\left[p, p+2\epsilon\right)$ and the total mixing cost and delivery cost would be $2C$. Therefore, for $\alpha$ matching, the total revenue would be $\$R\alpha$ and the total cost would be $2C\alpha$.

\vspace{0.2in}
\noindent
{\bf Only If part: }{\it If the GM instance a maximum revenue of $\$R\alpha$ and a minimum cost of $2C\alpha$ then the planar 3-DM has a matching of size $\alpha$.}

First, we argue that no two bins from the same set $X$ or $Y$ will be in the valid GM solution as the mixing cost between the same set of bins is set to \textit{infinity}. Therefore, there will be no profit from mixing from the same set of bins. Then, the same bin will not appear twice in the valid GM solution as it will either increase the overall mixing cost or decrease the overall revenue. For example, each bin has a capacity to hold half unit bushels and if a bin is delivered to multiple elevators the truck will be partially filled. For a partially filled truck the revenue would be $<\$R$ and there will be multiple mixing costs for the same bin which would be $>C$. Therefore, using the same bin to deliver grain to multiple elevators would either decrease the overall revenue or increase the overall cost. 

Finally, the elevator can accept only one unit of bushels. Therefore, if it accepts a fully loaded truck (truck capacity is also one unit bushels) then it will not accept bushels from any other truck. However, two partially loaded trucks may deliver to the same elevator which will violate it to be a valid matching. We argue that partially filled trucks would not be in the valid GM solution as in that case it will either increase the overall delivery cost or decrease the overall revenue. For example, if a partially loaded truck delivers to an elevator, the elevator will pay $<\$R$ amount and in that case, another partially filled truck has to deliver to the same elevator to get an overall revenue of $\$R$. However, in this case, multiple trucks have to go to the same elevator which will increase the delivery cost to be $>C$. Therefore, each truck in the valid GM solution will have a unique bin pair and an elevator providing a planar 3-DM of size $\alpha$. As GM is obviously in NP, we have the theorem.
\end{proof}

\begin{proof}
The second proof follows a reduction from the standard $3$-DM problem. Given an instance of the $3$-DM problem, we turn it into an instance of the $GM$ problem as follows:
\begin{enumerate}
    \item Create bin $b_{x_i}$ for each $x_i \in X$ and create bin $b_{y_j}$ for each $y_j \in Y$. Give all bins the capacity of one half unit.
    \item For each bin $b_a$, give it a unique protein content $p_a$ between $0$ and $1$ such that the average protein content of each pair of bins is unique. 
    This can be done by keeping a list of disallowed values including already assigned protein values, as well as protein values that when averaged to an already assigned value will equal some pair's average value.
    To select a new protein content value, a random value that is not on the disallowed value list is generated.
    \item Set all mixing costs initially to infinity.
    \item For each $z_k \in Z$, create the elevator $m_{z_k}$ and give it a capacity of one unit. Set all elevator cost functions initially to zero (i.e., elevators pay \$$0$ per unit, regardless of protein content).
    \item Create $|T|$ trucks, each having a capacity of one unit.
    \item Assign the following values:
    \begin{itemize}
        \item $\beta = \min \{|p_a - p_b|: a,b \text{ are different bins}\}$. $\beta > 0$ since the protein content of each bin is unique.
        \item $\delta = \min \{|\frac{p_a + p_b}{2} - \frac{p_c + p_d}{2}|: \text{ bins } \{a,b\} \neq \{c,d\} \}$. $\delta > 0$ since the average protein content for each pair of bins is unique.
        \item $\omega = \frac{2 \delta}{\beta} >= 2$ since $\delta=|\frac{p_a + p_b}{2} - \frac{p_c + p_d}{2}| = |\frac{p_a - p_c + p_b - p_d}{2}| \ge \frac{\beta + \beta}{2} = \beta$ implies $\frac{2 \delta}{\beta} \ge \frac{2 \delta}{\delta} = 2$.
    \end{itemize}
    \item Set delivery cost for each bin to $\omega$.
    \item For each $t = (x_i,y_j,z_k) \in T$:
    \begin{enumerate}
        \item Set mixing cost as $mix(b_{x_i}, b_{y_j})=0$.
        \item Modify $m_{z_k}$'s cost function to pay $(1+\omega)$ per unit for protein content of $\frac{p_{x_i} + p_{y_j}}{2}$, in addition to whatever the cost function already was.
    \end{enumerate}
\end{enumerate}
Figure~\ref{fig:reduct} shows the reduction from $3$-DM to $GM$. We then proof the following claim:  {\it The planar 3-DM problem has a matching of size $\alpha$ if and only if the GM instance has a profit of $\$\alpha$.}
\begin{figure}
\centering
\includegraphics[width = 0.9\linewidth]{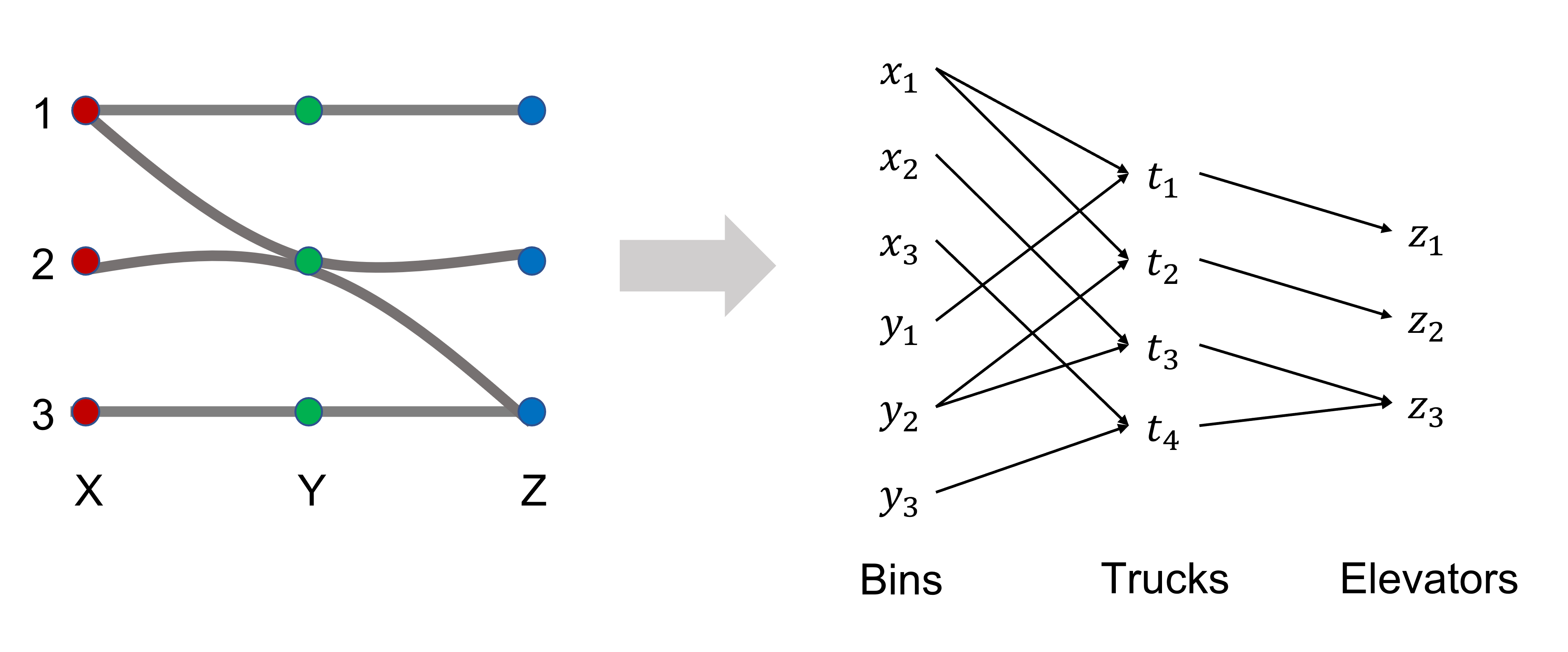}
\caption{Reduction from $3$-DM to the GM instance.}
\label{fig:reduct}
\end{figure}

\vspace{0.2in}
\noindent
{\bf If part: } Suppose the $3$-DM instance has a matching $M$ of size $\alpha$. 
For each $t = (x_i, y_j, z_k)$ in $M$, put one half unit from bin $b_{x_i}$ and one half unit from bin $b_{y_j}$ on truck $t$ and deliver to elevator $m_{z_k}$. 
Since each $x_i$, $y_j$, and $z_k$ can occur at most once each in $M$, the capacity constraints of bins $b_{x_i}$ and $b_{y_j}$, truck $r_{z_k}$, and elevator $m_{z_k}$ will not be violated. 
Also, since the protein content of truck $t$ is $(p_{x_i} + p_{y_j})/2$, elevator $m_{z_k}$ will pay $(1 + \omega)$ for its one unit. 
The profit of this assignment is $\sum_{t \in M}(1 + \omega - \omega)=\sum_{t \in M}(1)=\alpha$.

\vspace{0.2in}
\noindent
{\bf Only If part: } Suppose the $GM$ instance has an assignment whose profit is $\alpha$. 
Consider a elevator $m_{z_k}$ with profit larger than zero.
This elevator must have received material with a protein content matching some entry of its cost function.
In other words, the protein content of its received material must equal $(p_{x_i} + p_{y_j})/2$ for some $(x_i,y_j,z_k) \in T$. 
We will show that it is not possible for any bins other than $b_{x_i}$ and $b_{y_j}$ to have provided this material to $m_{z_k}$ at a profit.
Suppose that this material came from bins $b'_{x_i}$ and $b'_{y_j}$ instead of from $b_{x_i}$ and $b_{y_j}$.
We first aim to quantify the maximum amount of material that could be on the truck from $b'_{x_i}$ and $b'_{y_j}$ that achieved the protein content of $(p_{x_i} + p_{y_j})/2$.
Since the average protein content for each pair of bins is unique, this leads \hbox{to two cases:}

\underline{Case 1}: $\frac{p_{x_i} + p_{y_j}}{2} < \frac{p'_{x_i} + p'_{y_j}}{2}$

\noindent The maximum amount of material from bins $b'_{x_i}$ and $b'_{y_j}$ that yields a protein content of $(p_{x_i} + p_{y_j})/2$ must include the full half unit capacity of the bin with the lower protein content, say bin $b'_{x_i}$. Otherwise, the amount of additional material needed from bin $b'_{y_j}$ to maintain the weighted average protein content is less than the amount of material from left in bin $b'_{x_i}$.
The amount of material needed from bin $b'_{y_j}$ to achieve a protein content of $(p_{x_i} + p_{y_j})/2$ is represented as $q$ in the equation,
\begin{equation}
    \frac{\frac{1}{2} p'_{x_i} + q p'_{y_j}}{\frac{1}{2} + q} = \frac{p_{x_i} + p_{y_j}}{2}
\end{equation}
Furthermore, since the difference of average protein content between pairs of bins is at least $\delta$,
\begin{equation}
    \frac{p'_{x_i} + p'_{y_j}}{2} - \frac{p_{x_i} + p_{y_j}}{2} \ge \delta
\end{equation}
Combining these equations gives the bound of,
\begin{equation}
    \delta \le \frac{p'_{x_i} + p'_{y_j}}{2} - \frac{\frac{1}{2} p'_{x_i} + q p'_{y_j}}{\frac{1}{2} + q}
\end{equation}
Solving for $q$ gives the bound of,
\begin{align}
    \frac{\frac{1}{2} p'_{x_i} + q p'_{y_j}}{\frac{1}{2} + q} &\le \frac{p'_{x_i} + p'_{y_j}}{2} - \delta\\
    \frac{1}{2} p'_{x_i} + q p'_{y_j} &\le \frac{1}{2}\left(\frac{p'_{x_i} + p'_{y_j}}{2} - \delta \right) + q\left(\frac{p'_{x_i} + p'_{y_j}}{2} - \delta \right)\\
    q p'_{y_j} - q\left(\frac{p'_{x_i} + p'_{y_j}}{2} - \delta \right) &\le \frac{1}{2}\left(\frac{p'_{x_i} + p'_{y_j}}{2} - \delta - p'_{x_i}\right)\\
    q\left(\frac{p'_{y_j} - p'_{x_i}}{2} + \delta \right) &\le \frac{1}{2}\left(\frac{p'_{y_j} - p'_{x_i} }{2} - \delta \right)\\
    q&\le \frac{\frac{p'_{y_j} - p'_{x_i} }{2} - \delta}{2 \left(\frac{p'_{y_j} - p'_{x_i}}{2} + \delta \right)}
\end{align}
The total amount of material on the truck delivered to elevator $m_{z_k}$ can be calculated as,
\begin{align}
    q + \frac{1}{2} &\le \frac{\frac{p'_{y_j} - p'_{x_i} }{2} - \delta}{2 \left(\frac{p'_{y_j} - p'_{x_i}}{2} + \delta \right)} + \frac{1}{2}\\
    &= \frac{\frac{p'_{y_j} - p'_{x_i} }{2} - \delta}{2 \left(\frac{p'_{y_j} - p'_{x_i}}{2} + \delta \right)} + \frac{\frac{p'_{y_j} - p'_{x_i}}{2} + \delta}{2 \left(\frac{p'_{y_j} - p'_{x_i}}{2} + \delta \right)}\\
    &= \frac{p'_{y_j} - p'_{x_i}}{p'_{y_j} - p'_{x_i} + 2\delta}\\
    &= \left( \frac{p'_{y_j} - p'_{x_i} + 2\delta}{p'_{y_j} - p'_{x_i}}\right)^{-1}\\
    &= \left( 1 + \frac{2\delta}{p'_{y_j} - p'_{x_i}}\right)^{-1}\\ 
    &\le \left( 1 + \frac{2\delta}{\beta}\right)^{-1}\\
    &= \left( 1 + \omega\right)^{-1} = \frac{1}{1+\omega}
\end{align}
So, the truck delivers at most $1/(1+\omega)$ units to elevator $m_{z_k}$.
Since the elevator will pay $(1 + \omega)$ per unit and the delivery cost is $\omega$, the profit incurred is at most,
\begin{equation}
    \frac{1}{1+\omega} (1 + \omega) - \omega = 1 - \omega <= -1, \text{ since $\omega \ge 2$.}
\end{equation}
Thus, $m_{z_k}$ cannot profitably receive material with protein content $(p_{x_i} + p_{y_j})/2$ from any bins other than $b_{x_i}$ and $b_{y_j}$, which corresponds to an entry $(x_i, y_j, z_k)$ in $M$.

\underline{Case 2}: $\frac{p_{x_i} + p_{y_j}}{2} > \frac{p'_{x_i} + p'_{y_j}}{2}$

\noindent The maximum amount of material from bins $b'_{x_i}$ and $b'_{y_j}$ that yields a protein content of $(p_{x_i} + p_{y_j})/2$ must include the full half unit capacity of the bin with the higher protein content, say bin $b'_{y_j}$.
The amount of material needed from bin $b'_{x_i}$ to achieve a protein content of $(p_{x_i} + p_{y_j})/2$ is represented as $q$ in the equation,
\begin{equation}
    \frac{q p'_{x_i} + \frac{1}{2} p'_{y_j}}{\frac{1}{2} + q} = \frac{p_{x_i} + p_{y_j}}{2}
\end{equation}
Furthermore, since the difference of average protein content between pairs of bins is at least $\delta$,
\begin{equation}
    \frac{p_{x_i} + p_{y_j}}{2} - \frac{p'_{x_i} + p'_{y_j}}{2} \ge \delta
\end{equation}
Combining these equations gives the bound of,
\begin{equation}
    \delta \le \frac{q p'_{x_i} + \frac{1}{2} p'_{y_j}}{\frac{1}{2} + q} - \frac{p'_{x_i} + p'_{y_j}}{2}
\end{equation}
Solving for $q$ gives the bound of,
\begin{align}
    \delta + \frac{p'_{x_i} + p'_{y_j}}{2} &\le \frac{q p'_{x_i} + \frac{1}{2} p'_{y_j}}{\frac{1}{2} + q}\\
    \frac{1}{2}\left(\delta + \frac{p'_{x_i} + p'_{y_j}}{2}\right) + q\left(\delta + \frac{p'_{x_i} + p'_{y_j}}{2}\right) &\le q p'_{x_i} + \frac{1}{2} p'_{y_j}\\
    q\left(\delta + \frac{p'_{x_i} + p'_{y_j}}{2} - p'_{x_i}\right) &\le \frac{1}{2}\left(p'_{y_j} - \delta - \frac{p'_{x_i} + p'_{y_j}}{2}\right)\\
    q&\le \frac{p'_{y_j} - \delta - \frac{p'_{x_i} + p'_{y_j} }{2}}{2 \left(\delta + \frac{p'_{x_i} + p'_{y_j}}{2} - p'_{x_i}\right)}
\end{align}
The total amount of material on the truck delivered to elevator $m_{z_k}$ can be calculated as,
\begin{align}
    q + \frac{1}{2} &\le \frac{p'_{y_j} - \delta - \frac{p'_{x_i} + p'_{y_j} }{2}}{2 \left(\delta + \frac{p'_{x_i} + p'_{y_j}}{2} - p'_{x_i}\right)} + \frac{1}{2}\\
    &= \frac{p'_{y_j} - \delta - \frac{p'_{x_i} + p'_{y_j} }{2}}{2 \left(\delta + \frac{p'_{x_i} + p'_{y_j}}{2} - p'_{x_i}\right)} + \frac{\delta + \frac{p'_{x_i} + p'_{y_j}}{2} - p'_{x_i}}{2 \left(\delta + \frac{p'_{x_i} + p'_{y_j}}{2} - p'_{x_i}\right)}\\
    &= \frac{p'_{y_j} - p'_{x_i}}{p'_{y_j} - p'_{x_i} + 2\delta}\\
    &= \left( \frac{p'_{y_j} - p'_{x_i} + 2\delta}{p'_{y_j} - p'_{x_i}}\right)^{-1}\\
    &= \left( 1 + \frac{2\delta}{p'_{y_j} - p'_{x_i}}\right)^{-1}\\ 
    &\le \left( 1 + \frac{2\delta}{\beta}\right)^{-1}\\
    &= \left( 1 + \omega\right)^{-1} = \frac{1}{1+\omega}
\end{align}
So, the truck delivers at most $1/(1+\omega)$ units to elevator $m_{z_k}$.
Since the elevator will pay $(1 + \omega)$ per unit and the delivery cost is $\omega$, the profit incurred is at most,
\begin{equation}
    \frac{1}{1+\omega} (1 + \omega) - \omega = 1 - \omega <= -1, \text{ since $\omega \ge 2$.}
\end{equation}
Thus, $m_{z_k}$ cannot profitably receive material with protein content $(p_{x_i} + p_{y_j})/2$ from any bins other than $b_{x_i}$ and $b_{y_j}$, which corresponds to an entry $(x_i, y_j, z_k)$ in $M$.

In addition to pairs of bins not being able to profitably supply material at a protein content other than their average content, more than two bins can never be put on the same truck, since the mixing costs of bins from the same $X$ or $Y$ set is infinite.
Therefore, every profitable elevator $m_{z_k}$ must have received material only from the two bins that formed an entry $(x_i, y_j, z_k)$ in $M$.
Let $M$ be the set of all $(x_i, y_j, z_k)$ from each elevator with a profit larger than zero. 
We now must argue that there are no repeated $x$, $y$, or $z$ values in the set $M$:
\begin{itemize}
\item The profitability of elevator $z$ is calculated as $((1+\omega)q_z-\omega \gamma)$, where $q_z$ is the amount of material elevator $z$ receives and $\gamma 
\in \mathbb{Z}_{\ge 0}$ is the number of trucks that delivered to it. 
Since each elevator can accept up to one unit total, $q_z \le 1$. 
Thus, the profitability of an elevator is at most $(1+\omega - \omega \gamma)$, which is less than zero for any value of $\gamma$ larger than one since $\omega$ is at least two.
So, a single elevator cannot profitably receive from multiple trucks (i.e. pairs of bins) and there will be no repeated $z$ values in $M$.

\item A elevator that has been delivered a full unit must have come from the full capacity of a single pair of bins.
This means that bins that deliver to an elevator receiving a full unit cannot be used to deliver to multiple elevators since they all have a capacity of $\frac{1}{2}$.
Thus, if all elevators are delivered a full unit, there are no repeated $x$ or $y$ values in the set $M$.
However, it is possible for elevators to be profitable without receiving a full unit, so \textit{submaximal elevators} (i.e., profitable elevators that receive less than one unit) can exist in valid $GM$ solutions.

Submaximal elevators are a problem since they can lead to elevators sharing bins, which would make their corresponding $x$ and $y$ values appear multiple times in $M$ and prevent it from being a matching.
We proceed by first showing that submaximal elevators must share a bin with at most one other submaximal elevator.
We then show that we can turn a $GM$ solution with submaximal elevators into a valid matching by selecting either of the submaximal elevators in each pair, and that doing so will only increase the profit of the solution and value of the matching.

If an elevator $z$ is submaximal, both of its bins did not deliver their full $\frac{1}{2}$ capacity to that elevator, since the protein content would be invalid if a single bin delivered its full capacity and the other did not.
If both bins of a submaximal elevator $z$ deliver less than their $\frac{1}{2}$ capacity, then the bins cannot have been used to deliver to two other profitable elevators $z'$ and $z''$.
To be profitable, elevator $z$ must receive more than $\frac{2}{3}$ of a unit, since $\omega \ge 2$ and $((1 + \omega) q_z-\omega)\le0$ for all $q_z\le \frac{2}{3}$.
So, the most capacity that can remain in each bin of a profitable submaximal elevator that drew evenly from each bin is $\frac{1}{6} - \epsilon$.
Thus, if $t'$ and $t''$ draw the maximum capacity from their non shared bins, their delivered capacities are $\frac{1}{2} + \frac{1}{6} - \epsilon < \frac{2}{3}$.
Therefore, if both bins of a submaximal elevator $t$ deliver less than their $\frac{1}{2}$ capacity, then at most one of those bins can be used by another profitable submaximal elevator.
If that bin was not delivered to another elevator, then $z$ can easily be made maximal by delivering all of its bins' capacity to it.
If that bin did deliver to another profitable elevator $z'$, that elevator must also be submaximal.
The profit of the solution can be increased by making either of the submaximal elevators maximal and the other zero (for a combined profit of one), since the combined profit of the two submaximal elevators sharing a bin is at most $\frac{1}{2}$:
\begin{align}
    \nonumber (1 + \omega)q_z-\omega+(1+\omega)q_{z'}-\omega &= (1 + \omega)(q_z + q_{z'}) - 2\omega\\
    \nonumber &\le (1+\omega)\left(\frac{3}{2}\right)-2\omega \text{, since $q_z + q_{z'}\le \frac{3}{2}$}\\
    \nonumber &= \frac{3}{2} - \frac{1}{2} \omega \le \frac{1}{2} \text{, since $\omega \ge 2$}
\end{align}

This means that since submaximal elevators share bins with at most one other submaximal elevator, submaximal elevators are uniquely paired together.
Since turning paired submaximal elevators into a single maximal elevator increases the profit, we can ensure that there are no repeated $x$ or $y$ values in $M$ with a profit $\ge \alpha$.
Therefore, $M$ is a matching of the $3$-DM instance with size $\ge \alpha$.


\end{itemize}
\end{proof}

\section{Conclusion}
In this paper, we investigated the computational complexity of the grain mixing problem. We presented two proofs following a reduction from the $3$-DM and a planar $3$-DM problem respectively to show that finding the optimal solutions for the grain mixing problem is NP-Complete. This paper extends our previous work (\cite{noor_21}) where we used evolutionary approaches to find a quality solution for the grain mixing problem in a feasible time.  The proofs suggest that brute-force methods to find the optimal solution may be infeasible and justify the use of approximation algorithms (such as evolutionary approaches) to find a quality solution. 

\vskip 0.2in
\bibliography{sample}

\begin{thebibliography}{9}
\providecommand{\natexlab}[1]{#1}
\providecommand{\url}[1]{\texttt{#1}}
\expandafter\ifx\csname urlstyle\endcsname\relax
  \providecommand{\doi}[1]{doi: #1}\else
  \providecommand{\doi}{doi: \begingroup \urlstyle{rm}\Url}\fi

\bibitem[Bilgen and Ozkarahan(2007)]{MILP_blending_shipping}
Bilge Bilgen and Irem Ozkarahan.
\newblock A mixed-integer linear programming model for bulk grain blending and
  shipping.
\newblock \emph{International Journal of Production Economics}, 107\penalty0
  (2):\penalty0 555 -- 571, 2007.

\bibitem[Dyer and Frieze(1986)]{3dm_npc}
M~E Dyer and A~M Frieze.
\newblock Planar 3dm is \textit{NP}-complete.
\newblock \emph{J. Algorithms}, 7\penalty0 (2):\penalty0 174–184, 1986.

\bibitem[Garey and Johnson(1990)]{complexity_90}
Michael~R. Garey and David~S. Johnson.
\newblock \emph{Computers and Intractability; A Guide to the Theory of
  NP-Completeness}.
\newblock W. H. Freeman and Co., 1990.

\bibitem[Haas(2011)]{blending_thesis}
Nikolas Haas.
\newblock \emph{Optimizing Wheat Blends for Customer Value Creation: A Special
  Case of Solvent Retention Capacity}.
\newblock MS Thesis, Kansas State University, USA, 2011.

\bibitem[Hayta and Cakmalki(2001)]{blending_2001}
Mehmet Hayta and Unsal Cakmalki.
\newblock Optimization of wheat blending to produce breadmaking flour.
\newblock \emph{Journal of Food Process Engineering}, 24:\penalty0 179 -- 192,
  08 2001.

\bibitem[Karp(1972)]{Karp1972}
Richard~M. Karp.
\newblock \emph{Reducibility among Combinatorial Problems}, pages 85--103.
\newblock Springer US, Boston, MA, 1972.

\bibitem[Krentel(1988)]{milp_hardness_88}
Mark~W. Krentel.
\newblock The complexity of optimization problems.
\newblock \emph{Journal of Computer and System Sciences}, 36\penalty0
  (3):\penalty0 490--509, 1988.

\bibitem[Li et~al.(2014)Li, Bonyadi, Michalewicz, and
  Barone]{blending_hybrid_evo}
Xiang Li, Mohammad~reza Bonyadi, Zbigniew Michalewicz, and Luigi Barone.
\newblock A hybrid evolutionary algorithm for wheat blending problem.
\newblock \emph{TheScientificWorldJournal}, 2014:\penalty0 967254, 02 2014.

\bibitem[Noor and Sheppard(2021)]{noor_21}
Md~Asaduzzaman Noor and John~W. Sheppard.
\newblock Evolutionary grain-mixing to improve profitability in farming winter
  wheat.
\newblock In \emph{Applications of Evolutionary Computation}, pages 113--129.
  Springer International Publishing, 2021.

\end{thebibliography}

\end{document}